\documentclass{IEEEtran}
\usepackage{cite}
\usepackage{amsmath,amssymb,amsthm,amsfonts}
\usepackage{algorithmic}
\usepackage{graphicx}
\usepackage{textcomp}
\usepackage{subfigure}
\usepackage{multirow}
\usepackage{booktabs}
\usepackage{diagbox}
\usepackage{autobreak}
\usepackage{booktabs}
\usepackage{footmisc}
\usepackage{stfloats}
\usepackage{threeparttable}

\begin{document}
\title{Outage Analysis for Intelligent Reflecting Surface Assisted Vehicular Communication Networks}
\author{
\IEEEauthorblockN{Jue Wang\IEEEauthorrefmark{1}, Wence Zhang\IEEEauthorrefmark{2}, Xu Bao\IEEEauthorrefmark{2}, Tiecheng Song \IEEEauthorrefmark{1}, Cunhua Pan\IEEEauthorrefmark{3}
}

\IEEEauthorblockA{\IEEEauthorrefmark{1}National Mobile Communication Research Laboratory, Southeast University, Nanjing, China}

\IEEEauthorblockA{\IEEEauthorrefmark{2}School of Computer Science and Communication Engineering, Jiangsu University, Zhenjiang, China}

\IEEEauthorblockA{\IEEEauthorrefmark{3}School of Electronic Engineering and Computer Science, Queen Mary University of London, London, U.K.}
\\ e-mail: \{{220180893, songtc}\}@seu.edu.cn, \{{wencezhang, xbao}\}@ujs.edu.cn, c.pan@qmul.ac.uk
}

\maketitle
\pagestyle{empty}
\thispagestyle{empty}

\begin{abstract}
Vehicular communication is an important application of the fifth generation of mobile communication systems (5G). Due to its low cost and energy efficiency, intelligent reflecting surface (IRS) has been envisioned as a promising technique that can enhance the coverage performance significantly by passive beamforming. In this paper, we analyze the outage probability performance in IRS-assisted vehicular communication networks. We derive the expression of outage probability by utilizing series expansion and central limit theorem. Numerical results show that the IRS can significantly reduce the outage probability for vehicles in its vicinity. The outage probability is closely related to the vehicle density and the number of IRS elements, and better performance is achieved with more reflecting elements.
\end{abstract}


\section{Introduction}
Vehicular communication networks have been extensively investigated to realize the concept of intelligent transportation systems (ITS) \cite{8594703}. As an important application of the fifth generation of mobile communication systems (5G), vehicular communication networks require high data-rate, low-latency transmission and high reliability. Traditional vehicular communication has been supported by the dedicated short-range communication (DSRC) standard. However, DSRC may not effectively meet the requirements of high data rate due to its limited coverage and capacity\cite{8539687}.

Therefore, to improve the data rate, the millimeter wave (mmWave) communications have been introduced into vehicular communication networks. With larger available bandwidth, mmWave is able to support the high data rate and low latency communications for the emerging vehicular applications\cite{8594703}. Nevertheless, due to its various problems such as high path loss and vulnerability to blockage\cite{8617303}, the beam-tracking and transmission outage remain a challenge for mmWave communications\cite{8594703}. Therefore, high outage probability and low coverage rate are still unsolved problems for mmWave vehicular communications.

Among the latest technologies for green and effective wireless communication, the intelligent reflecting surface (IRS) has emerged as an innovative and cost-effective paradigm for improving the transmission coverage and signal quality via passive reflecting arrays\cite{pan2019multicell}. Owing to the tunable phase shifts of all reflecting elements, signal enhancement and interference suppression can be achieved by the IRS without the use of active transmitters\cite{pan2019multicell}.
The authors of\cite{8796365} provide a detailed overview and historical perspective on state-of-the-art solutions of IRS-aided communication systems. Various aspects of IRS have recently been investigated in terms of channel estimation\cite{8879620,03272}, data rate maximization\cite{pan2019inte,8683145}, physical layer security\cite{hong2019}, etc.

To the best of the authors' knowledge, there are a paucity of contributions studying the integration of IRS in vehicular network systems. In\cite{12183}, the physical layer security of vehicular networks employing IRS was studied and two IRS-based vehicular network system models were designed. However, the work only focused on the secrecy capacity and the outage performance of IRS-assisted vehicular networks has not been studied. Given the blockage and vehicle density in practical road conditions, the difficulty of analysis would increase, which were not considered in\cite{12183}.

In this paper, we propose to use IRS in vehicular networks to improve the coverage and outage performance. We establish a model for an IRS-assisted downlink vehicular communication system with Poisson Point Process (PPP) distribution and derive the expression of outage probability based on series expansion and central limit theorem. The outage performance of traditional vehicular networks and IRS-assisted vehicular networks are also compared.

The main contributions of this work are listed as follows.
\begin{itemize}
\item We establish an IRS-assisted downlink vehicular communication model by considering a more practical scenario with vehicle blockings and density of vehicles.
\item Accurate approximations of the outage probability are derived based on series expansion and central limit theorem. Lower vehicle density and larger size of IRS will improve the outage probability performance.
\item Simulation analysis is conducted and it is shown that the outage probability can be reduced significantly by using large IRS. Besides, IRS can effectively improve the coverage performance of vehicular communication networks.
\end{itemize}

The rest of this paper is summarized as follows. In Section II, the IRS-assisted downlink channel model is given and the problem for outage analysis is formulated. In Section III, the expression for outage probability is derived based on series expansion and central limit theorem. The numerical results are presented in Section VI and conclusions are drawn in Section V.

\section{System Model and Problem Formulation}
In this paper, we consider an IRS-assisted downlink vehicular communication system, as illustrated in Fig. 1.
Each direction of the road includes two lanes.
A single-antenna roadside unit (RSU) is deployed on the separating strip between lanes of the two directions.
The single-antenna target user locates at the outer lane and the direct link between the RSU and the target user may be blocked by vehicles at the inner lane.
To assist the communication between the RSU and the target user, an IRS consisting of $M$ passive reflecting elements is deployed at the other roadside.
\begin{figure}
    \centering
    \includegraphics[scale=0.40]{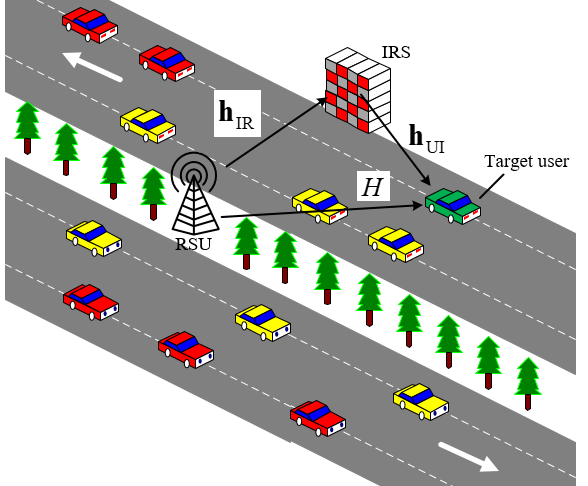}
    \caption{An illustration of IRS-assisted vehicle communication systems}
\end{figure}

The channel of the RSU-user link, RSU-IRS link and IRS-user link are denoted by $H\in\mathbb C$, $\mathbf{h} _\mathrm{IR}\in\mathbb C^{M \times 1}$ and $\mathbf{h} _\mathrm{UI}\in\mathbb C^{M \times 1}$, respectively. Since the locations of obstacle vehicles are changing all the time, the RSU-user and the RSU-IRS links may be blocked. Therefore, $H$ and $\mathbf{h} _\mathrm{IR}$ may experience different large-scale fading.

\subsection{Channel Model}

In order to take into account the blocking that exists in RSU-user and RSU-IRS links, we assume that the shape of vehicles are rectangle and the length of each vehicle is $\tau$. When the direct link connecting the transmitter and the receiver passes through obstacle vehicles, the channel is blocked. Hence, to avoid channel blocking, no vehicle should exist in the vicinity of the direct link with a range of $\frac{\tau }{2}$, as depicted in Fig. 2. When blocking happens, it indicates an NLOS channel between the transmitter and the receiver. Otherwise, it is a LOS channel.
\begin{figure}
    \centering
    \includegraphics[scale=0.45]{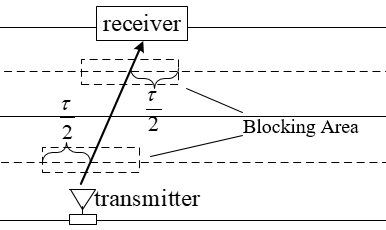}
    \caption{An illustration of channel blocking}
\end{figure}

Poisson Point Process (PPP) is very useful for modeling the distribution of vehicles on a lane\cite{TassiModeling}. We use two PPPs, i.e., $\mathrm{\Phi _O}$ with density $\lambda _\mathrm O$ and $\mathrm{\Phi _U}$ with density $\lambda _\mathrm U$, to model the vehicle distribution in the obstacle lane and the user lane, respectively. The density value is related to the average speed and density of vehicles. For a PPP with density $\lambda$, let us denote $N\left( {\left[ {a,b} \right)} \right)$ as the number of points (events) within the interval of $\left[ {a,b} \right)$, the distribution of which is derived as
\begin{equation}
  \mathrm P\left[ {N\left( {\left[ {a,b} \right)} \right){\rm{ = }}k} \right] = {\mathrm e^{ - \lambda \left( {b - a} \right)}}\frac{{{{\left[ {\lambda \left( {b - a} \right)} \right]}^k}}}{{k!}}.
\end{equation}

The probabilities of the RSU-IRS channel being unblocked (LOS) and blocked (NLOS) are denoted by ${p_\mathrm{IR,L}}$ and ${p_\mathrm{IR,N}} = 1 - {p_\mathrm{IR,L}}$, respectively, while for the RSU-user channel, the probabilities are denoted by ${p_\mathrm{UR,L}}$ and ${p_\mathrm{UR,N}} = 1 - {p_\mathrm{IR,L}}$, respectively. By some mathematical manipulations, we obtain
\begin{subequations}
\begin{equation}
  {p_\mathrm{IR,L}} = {\mathrm e^{ - {\lambda _\mathrm{O}}\tau }}{\mathrm e^{ - {\lambda _\mathrm{U}}\tau }}, \ {p_\mathrm{IR,N}} = 1 - {\mathrm e^{ - {\lambda _\mathrm{O}}\tau }}{\mathrm e^{ - {\lambda _\mathrm{U}}\tau }};
\end{equation}
\begin{equation}
  {p_\mathrm{UR,L}} = {\mathrm e^{ - {\lambda _\mathrm{O}}\tau }}, \ {p_\mathrm{UR,N}} = 1 - {\mathrm e^{ - {\lambda _\mathrm{O}}\tau }}.
\end{equation}
\end{subequations}

NLOS and LOS scenarios result in different path loss and small-scale fading. The path loss of a specific link at a distance of $r$ is modeled as
\begin{equation}
  l(r) = {{\delta} _\mathrm L}{C_\mathrm L}{r^{ - {\alpha(z) _\mathrm L}}} + (1 - {{\delta} _\mathrm L}){C_\mathrm N}{r^{ - {\alpha(z) _\mathrm N}}},
\end{equation}
where ${\alpha(z) _\mathrm L}$ and ${\alpha(z) _\mathrm N}$ are the pass loss exponents, while ${C _\mathrm L}$ and ${C _\mathrm N}$ are the path loss intercept factors of the LOS and NLOS cases, respectively. In (3), $\delta _\mathrm L$ is equal to one for LOS and zero for NLOS.

Regarding small-scale fading, Rayleigh and Rician fading channels are used for modeling NLOS and LOS cases, respectively.
The probability density function (PDF) of Rayleigh and Rician distribution are given, respectively, by
\begin{subequations}
\begin{equation}
  p_\mathrm{ray}(x){\rm{ = }}\left\{ {\begin{array}{*{20}{c}}
{\frac{x}{{\sigma _0^2}}{\mathrm e^{ - \frac{{{x^2}}}{{2\sigma _0^2}}}},{\kern 1pt} {\kern 1pt} {\kern 1pt} {\kern 1pt} {\kern 1pt} {\kern 1pt} {\kern 1pt} {\kern 1pt} {\kern 1pt} {\kern 1pt} x > 0}\\
{0,{\kern 1pt} {\kern 1pt} {\kern 1pt} {\kern 1pt} {\kern 1pt} {\kern 1pt} {\kern 1pt} {\kern 1pt} {\kern 1pt} {\kern 1pt} {\kern 1pt} {\kern 1pt} {\kern 1pt} {\kern 1pt} {\kern 1pt} {\kern 1pt} {\kern 1pt} {\kern 1pt} {\kern 1pt} {\kern 1pt} {\kern 1pt} {\kern 1pt} {\rm{otherwise}}}
\end{array}} \right.,
\end{equation}
\begin{equation}
  p_\mathrm{ri}(x){\rm{ = \!}}\left\{ {\begin{array}{*{20}{c}}
{p(x) = \frac{x}{{\sigma _m^2}}\exp\! \left( { - \frac{{{x^2} + {A^2}}}{{2\sigma _m^2}}} \right){I_0}\left( {\frac{{xA}}{{\sigma _m^2}}} \right),{\kern 1pt} {\kern 1pt} {\kern 1pt} {\kern 1pt} {\kern 1pt} {\kern 1pt} {\kern 1pt} {\kern 1pt} {\kern 1pt} {\kern 1pt} x > 0}\\
{0,{\kern 1pt} {\kern 1pt} {\kern 1pt} {\kern 1pt} {\kern 1pt} {\kern 1pt} {\kern 1pt} {\kern 1pt} {\kern 1pt} {\kern 1pt} {\kern 1pt} {\kern 1pt} {\kern 1pt} {\kern 1pt} {\kern 1pt} {\kern 1pt} {\kern 1pt} {\kern 1pt} {\kern 1pt} {\kern 1pt} {\kern 1pt} {\kern 1pt} {\kern 1pt} {\kern 1pt} {\kern 1pt} {\kern 1pt} {\kern 1pt} {\kern 1pt} {\kern 1pt} {\kern 1pt} {\kern 1pt} {\kern 1pt} {\kern 1pt} {\kern 1pt} {\kern 1pt} {\kern 1pt} {\kern 1pt} {\kern 1pt} {\kern 1pt} {\kern 1pt} {\kern 1pt} {\kern 1pt} {\kern 1pt} {\kern 1pt} {\kern 1pt} {\kern 1pt} {\kern 1pt} {\kern 1pt} {\kern 1pt} {\kern 1pt} {\kern 1pt} {\kern 1pt} {\kern 1pt} {\kern 1pt} {\kern 1pt} {\kern 1pt} {\kern 1pt} {\kern 1pt} {\kern 1pt} {\kern 1pt} {\kern 1pt} {\kern 1pt} {\kern 1pt} {\kern 1pt} {\kern 1pt} {\kern 1pt} {\kern 1pt} {\kern 1pt} {\kern 1pt} {\kern 1pt} {\kern 1pt} {\kern 1pt} {\kern 1pt} {\kern 1pt} {\kern 1pt} {\kern 1pt} {\kern 1pt} {\kern 1pt} {\kern 1pt} {\kern 1pt} {\kern 1pt} {\kern 1pt} {\kern 1pt} {\kern 1pt} {\kern 1pt} {\kern 1pt} {\kern 1pt} {\kern 1pt} {\kern 1pt} {\kern 1pt} {\kern 1pt} {\kern 1pt} {\kern 1pt} {\kern 1pt} {\kern 1pt} {\kern 1pt} {\kern 1pt} {\kern 1pt} {\kern 1pt} {\kern 1pt} {\kern 1pt} {\kern 1pt} {\kern 1pt} {\kern 1pt} {\kern 1pt} {\kern 1pt} {\kern 1pt} {\kern 1pt} {\kern 1pt} {\kern 1pt} {\kern 1pt} {\kern 1pt} {\kern 1pt} {\kern 1pt} {\kern 1pt} {\kern 1pt} {\kern 1pt} {\kern 1pt} {\kern 1pt} {\kern 1pt} {\kern 1pt} {\kern 1pt} {\kern 1pt} {\kern 1pt} {\kern 1pt} {\kern 1pt} {\kern 1pt} {\kern 1pt} {\kern 1pt} {\kern 1pt} {\kern 1pt} {\kern 1pt} {\kern 1pt} {\kern 1pt} {\kern 1pt} {\kern 1pt} {\kern 1pt} {\kern 1pt} {\kern 1pt} {\kern 1pt} \rm{otherwise}}
\end{array}} \right.,
\end{equation}
\end{subequations}
where $\sigma _0^2$ denotes the variance of the components in a Rayleigh variable; $A$ denotes the amplitude of the main signal component of a Rician variable; $\sigma _\mathrm m^2$ denotes the power of the multi-path signal component of a Rician variable; ${I_0}\left( \cdot \right)$ is the modified zero-order Bessel function.

\subsection{Problem Formulation}

Assume that the transmitted symbol $x$ is expressed as
\begin{equation}
  x = \sqrt {{P_\mathrm{tx}}} s,
\end{equation}
where $P_\mathrm{tx}$ is the transmission power, $s$ is the information-carrying symbol which follows the distribution of $C{\cal N}(0,1)$. Hence, the received signal at the target user can be written as
\begin{equation}
  y = \left[ {H\sqrt {l(r_\mathrm{UR})}  + \mathbf h_\mathrm{UI}^\mathrm T\mathbf\Theta {\mathbf h_\mathrm{IR}}\sqrt {l({r_\mathrm{IR}})l({r_\mathrm{UI}})} } \right]x + n,
\end{equation}
where $\mathbf\Theta  = \mathrm{diag}({\mathrm e^{ - \mathrm j{\theta _1}}},{\mathrm e^{ - j{\theta _2}}}, \ldots ,{\mathrm e^{ - \mathrm j{\theta _M}}})$ contains the phase shifting coefficients of the IRS; $n$ denotes the received white Gaussian noise which follows the distribution of $C{\cal N}(0,\sigma _\mathrm{n}^2)$; $r_\mathrm{UR}$, $r_\mathrm{IR}$ and $r_\mathrm{UI}$ denote the distance of RSU-user, RSU-IRS and IRS-User links, respectively.

According to (6), the SNR of the received signal is given by
\begin{equation}
  \begin{array}{l}
{\rm SNR}{\mathrm{ = }}\frac{{{P_{\rm tx}}}}{{\sigma _\mathrm n^2}}{\left| {H\sqrt {l(r)}  + \mathbf h_\mathrm{UI}^\mathrm T\mathbf \Theta {\mathbf h_\mathrm{IR}}\sqrt {l({r_\mathrm{IR}})l({r_\mathrm{UI}})} } \right|^2}.\\
\end{array}
\end{equation}

Note that the ${\rm SNR}$ will be different when $l(r)$ and $l(r_\mathrm{UI})$ change under different road conditions. For simplicity, we denote ${L_\mathrm{D}} = \sqrt {l(r)}$ and ${L_\mathrm R} = \sqrt {l({r_\mathrm{IR}})l({r_\mathrm{UI}})}$.

The channel coefficients can be expressed in detail as
\begin{subequations}
\begin{equation}
  H = |H|{\mathrm e^{\mathrm{j}\varphi }},
\end{equation}
\begin{equation}
  {\mathbf h_\mathrm{UI}} = {\left[ {|{\mathrm h_\mathrm{UI,1}}|{\mathrm e^{\mathrm{j}{\varphi _\mathrm{UI,1}}}},|{\mathrm h_\mathrm{UI,2}}|{\mathrm e^{\mathrm{j}{\varphi _\mathrm{UI,2}}}}, \ldots ,|{\mathrm h_{\mathrm {UI},M}}|{\mathrm e^{\mathrm{j}{\varphi _{\mathrm {UI},M}}}}} \right]^\mathrm T},
\end{equation}
\begin{equation}
  {\mathbf h_\mathrm{IR}} = {\left[ {|{\mathrm h_\mathrm{IR,1}}|{\mathrm e^{\mathrm{j}{\varphi _\mathrm{IR,1}}}},|{\mathrm h_\mathrm{IR,2}}|{\mathrm e^{\mathrm{j}{\varphi _\mathrm{IR,2}}}}, \ldots ,|{\mathrm h_{\mathrm {IR},M}}|{\mathrm e^{\mathrm{j}{\varphi _{\mathrm {IR},M}}}}} \right]^\mathrm T}.
\end{equation}
\end{subequations}

According to \cite{8796365}, the SNR is maximized when the phase shift of the direct and the reflection paths are the same, which means the optimal phase shift is given by ${\theta _i} = {\varphi _{\mathrm{UI},i}} + {\varphi _{\mathrm{IR},i}} - \varphi $. Therefore, (7) becomes
\begin{equation}
  \mathrm{SNR}{\rm{ = }}\frac{{{P_\mathrm{tx}}}}{{\sigma _\mathrm n^2}}{\left| {{L_\mathrm{D}}|H| + {L_\mathrm{R}} \times \sum\limits_{i = 1}^M {|{h_{\mathrm{UI},i}}||{h_{\mathrm{IR},i}}|} } \right|^2}.
\end{equation}

The outage probability of a given threshold $t$ is given by
\begin{equation}
  \mathrm{P_O}(t ){\rm{ = }}\mathrm P(\mathrm{SNR} < t ).
\end{equation}

In (10), the SNR is a function of the channel coefficients between the RSU and the user. Due to the vehicle distribution and the complicated channel distribution which involves both Rayleigh and Rician distributions, a closed-form expression for (10) is very difficult to obtain.

\section{Outage Probability Analysis}

To simplify the analysis of outage probability, let us first examine the summation term $ G\triangleq\sum\nolimits_{i = 1}^M {|{h_{\mathrm {UI},i}}||{h_{\mathrm{IR},i}}|}$ in (9). For a specific $i$, denote the mean and variance of $\zeta_i=|{h_{\mathrm {UI},i}}||{h_{\mathrm{IR},i}}|$ as ${{\cal E}_0}$ and ${{\cal V}_0}$, respectively.
According to CLT, since $\zeta_i$'s are independently identically distributed (i.i.d), we have $ G\sim {\cal N}(M{\cal E}_0,M{\cal V}_0)$ when $M$ is large.
Substituting (9) into (10) yields
\begin{equation}
  \begin{split}
{\mathrm P_\mathrm O}(t )&{\rm{ = }}\mathrm P\left( {\mathrm{SNR} < t } \right)\\
&{\rm{ = }}\mathrm P\left( {\frac{{{P_\mathrm{tx}}}}{{\sigma _\mathrm n^2}}{{\left| {{L_\mathrm D}|H| + {L_\mathrm R} G } \right|}^2} < t } \right)\\
&{\rm{=}}\mathrm P\left( {{L_\mathrm D}|H| + {L_\mathrm R}G < Z} \right).
\end{split}
\end{equation}
where $Z\rm{ = }{\sigma _\mathrm n}\sqrt {\frac{ \emph t }{{{P_\mathrm {tx}}}}}$ is the modified threshold. In the following derivation, we use $\mathrm P_\mathrm O(Z)$ to represent the outage probability.

The analysis focuses on deriving the distribution of the random variable ${Z'}={L_\mathrm D}|H| + {L_\mathrm R}G$ in (11). Then the outage probability can be obtained by integration. Since a closed-form of (11) is very difficult to derive, we provide accurate approximations based on series expansion based approximation (SEA) and central limit theorem based approximation (CLA).

\subsection{Analysis based on SEA}
The distribution of ${Z'}$ varies in different channel conditions. Both the RSU-user link and RSU-IRS link could be LOS or NLOS channels. We consider the following four situations.

\subsubsection{NLOS for both RSU-user and RSU-IRS links}

\newtheorem{thm}{Theorem}
\newtheorem{prop}{Proposition}

Since the RSU-user and RSU-IRS links are both NLOS channels, both $|H|$ and $|{h_{\mathrm{IR},i}}|$ follow Rayleigh distribution in (4a). Note that $|{h_{\mathrm{UI},i}}|$ is always a Rician distributed variable, which follows the distribution in (4b). In this case, denote the mean and variance of $G$ as ${\cal E}_\mathrm N$ and ${\cal V}_\mathrm N$, which are given by
\begin{subequations}
\begin{equation}
{\cal E}_\mathrm N\! = M{\cal E}_0
\!=\!\frac{\pi }{2}M {\sigma _0}{\sigma _\mathrm m}{\mathrm e^{ - \frac{K}{2}}}\!\left[ {(1\! +\! K){I_0}\!\left( {\frac{K}{2}} \right)\! +\! K{I_1}\!\left( {\frac{K}{2}} \right)} \right],
\end{equation}
\begin{equation}
{\cal V}_\mathrm N = M{\cal E}_0 = M \left({4\sigma _0^2\sigma _\mathrm m^2 + 2\sigma _0^2{A^2} - {\cal E}_\mathrm N^2}\right),
\end{equation}
\end{subequations}
where $K{\rm{ = }}{A^2}/({2\sigma _{\rm{m}}^2})$ and $I_1(\cdot)$ is the modified first-order Bessel function.

The PDF of $Z'$ is obtained after some mathematical manipulations as
\begin{equation}
\begin{split}
&f_\mathrm{NN}\left( z \right) =  \frac{1}{{\sqrt {2\pi M{{\cal V}_\mathrm N}} {L_\mathrm R}\sigma _0^2}} \exp\left({{-\frac{\alpha(z)}{2{L_\mathrm D}\sigma _0^2}}}\right) \times \\
&\left\{ {\frac{1}{2}\alpha(z) \sqrt {\frac{\pi }{{\beta}_\mathrm N }} \left[ { 1 +  {\rm{erf}}\left( {\sqrt {\beta_\mathrm N}  \alpha(z) } \right)} \right] +
 \frac{1}{{2{\beta}_\mathrm N }}\mathrm e^{- \beta_\mathrm N {\alpha(z) ^2}}} \right\},
\end{split}
\end{equation}
where $\alpha(z)=\frac{{{L_\mathrm D}\sigma _0^2\left( {z - {L_\mathrm R}M{{\cal E}_\mathrm N}} \right)}}{{M{{\cal V}_\mathrm N}L_\mathrm R^2 + L_\mathrm D^2\sigma _0^2}}$ and ${\beta}_\mathrm N=\frac{{M{{\cal V}_\mathrm N}L_\mathrm R^2 + L_\mathrm D^2\sigma _0^2}}{{2M{{\cal V}_\mathrm N}L_\mathrm R^2\sigma _0^2}}$.

Note that $\mathrm P_\mathrm O(Z)$ is obtained by integration of (13) over $z$. However, the integral is very difficult since (13) contains error function $\rm{erf}(\cdot)$ to derive in closed form. In this work, we provide series expansion based approximation to solve the integral. In (13), $\mathrm{erf}(x)$ can be expressed in the following series as
\begin{equation}
  \begin{split}
{\rm{erf}}\left( x \right)  \approx \frac{2}{{\sqrt \pi  }}\sum\limits_{i = 0}^n  {\frac{{{{\left( { - 1} \right)}^i}}}{{(2i + 1) \cdot i!}}{x^{2i + 1}}}
\end{split}
\end{equation}
where $n$ is the series order.

Applying (14) to (13), we obtain the expression of $\mathrm P_\mathrm O(Z)$.

\begin{figure*}[hb]
\hrulefill
\begin{equation}
  {C_{\mathrm L,ij}} = \frac{1}{{\sqrt {2\pi M{{\cal V}_\mathrm L}} }}\frac{1}{{\sigma _\mathrm m^2}}\mathrm{exp}\left( { - \frac{{{A^2}}}{{2\sigma _\mathrm m^2}}} \right) \frac{1}{{{2^{i + 2j}}\sigma _\mathrm m^{2i + 4j}}}\frac{{{{\left( { - 1} \right)}^i}{A^{2j}}}}{{i!  {{\left( {j!} \right)}^2}}}\frac{1}{{2i + 2j + 2}}{\left( {\frac{{{L_\mathrm R}}}{{{L_\mathrm D}}}} \right)^{2i + 2j + 2}}.
\end{equation}
\begin{equation}
F_\text L(x)=\sum\limits_{i = 0}^n {\frac{{{{\left( { - 1} \right)}^i}{\beta_\mathrm L ^i}}}{{(2i + 1) \cdot i!}}\frac{{{L_\mathrm D}^{2i + 2}\sigma _0^{4i + 4}{{2}^{i + 0.5}}}}{{{\left({M{{\cal V}_\mathrm L}L_\mathrm R^2 + L_\mathrm D^2}\right) ^{i + 0.5}}}}\Gamma \left( {i + 1.5, x} \right)}.
\end{equation}
\end{figure*}

\begin{prop}
When the RSU-user link and the RSU-IRS link are both NLOS, the outage probability can be expressed as
\begin{equation}
  \begin{array}{l}
{\mathrm P_{\mathrm O,\text{NN}}}(Z){\rm{ = }} C({\cal V}_\mathrm N)\times \\\left\{ {\frac{\sqrt{\pi}}{{4{\beta}_\mathrm N }D_\mathrm N} \left[ {{\rm{erf}}\left( {D_\mathrm N \left( {Z - {B_\mathrm N}} \right)} \right) + }\right. {\left.{{\rm{erf}}\left( {D_\mathrm N {B_\mathrm N}} \right)} \right]} - }\right.\\
 \left.{\frac{1}{2}\sqrt {\frac{\pi }{{\beta}_\mathrm N }} {L_\mathrm D}\sigma _0^2\left( {\text e ^{ -a_\mathrm N} \!- \!\text e ^{ -b_\mathrm N}} \right) + 2F_\text N(0) \! - \!F_\text N(b_\mathrm N) \! - \! F_\text N(a_\mathrm N)}\right\},
\end{array}
\end{equation}
where $B_\mathrm N={L_\mathrm R}M{{\cal E}_\mathrm N}$, $D_\mathrm N\!=\!\sqrt {(A_\mathrm N \! + \!2{\beta}_\mathrm N L_\mathrm D^2\sigma _0^4)/(2{A_\mathrm N ^2})}$,\\
$b_\mathrm N\!=\!B_\mathrm N^2/(2A_\mathrm N)$ \  and \
$a_\mathrm N\!=\!\left( {Z\! -\! {B_\mathrm N}} \right)^2\!/(2A_\mathrm N)$ \ with \\$A_\mathrm N\!=\!M{{\cal V}_\mathrm N}L_\mathrm R^2\!+\! L_\mathrm D^2\sigma _0^2$. \ $C(x)\!=\!\left({\!\sqrt {2\pi M{x}} {L_\mathrm R}\sigma _0^2\!}\right)^{\!-1}$ and \\$F_\text N(x)=\sum\limits_{i = 0}^n {\frac{{{{\left( { - 1} \right)}^i}{{\beta}_\mathrm N ^i}}}{{(2i + 1) \cdot i!}}\frac{{{L_\mathrm D}^{2i + 2}\sigma _0^{4i + 4}{{2}^{i + 0.5}}}}{{{\left({M{{\cal V}_\mathrm N}L_\mathrm R^2 + L_\mathrm D^2}\right) ^{i + 0.5}}}}\Gamma \left( {i + 1.5, x} \right)}$ \\
with $\Gamma \left( {\alpha(z) ,x} \right) = \int_x^\infty  {{\mathrm e^{ - t}}{t^{\alpha(z)  - 1}}\mathrm dt}$.
\begin{proof}
In (13), supposing $z - {B_\mathrm N} = t$, the outage probability can be rewritten as
\begin{equation}
\begin{array}{l}
{\mathrm P_\mathrm{O,NN}}(Z ) {\rm{ = }}\!\int_{ - {B_\mathrm N}}^{Z - {B_\mathrm N}} \!{C({\cal V}_\mathrm N)\exp \left( { \!-\! \frac{{{t^2}}}{{2A_\mathrm N }}} \right) } \!\left[ {\frac{1}{{2\beta }}\exp \left( { - \beta \frac{{L_\mathrm D^2\sigma _0^4{t^2}}}{{{A_\mathrm N ^2}}}} \right) }\right. \\
+\frac{1}{2}\sqrt {\frac{\pi }{\beta }} \frac{{{L_\mathrm D}\sigma _0^2t}}{A_\mathrm N } + \left.{\sum\limits_{i = 0}^\infty  {\frac{{{{\left( { - 1} \right)}^i}}}{{(2i + 1) \cdot i!}}{\beta ^i}\frac{{{L_\mathrm D}^{2i + 2}\sigma _0^{4i + 4}{t^{2i + 2}}}}{{{A_\mathrm N ^{2i + 2}}}}} } \right]\mathrm dt\\
{\rm{ = }} C({\cal V}_\mathrm N)\!\left[{ \left({\frac{1}{{4{\beta}_\mathrm N }}\frac{\sqrt{\pi}}{D_\mathrm N} {\rm{erf}}\left( {D_\mathrm N t} \right) \!-\!\! }\right. \left.{{\sqrt {\frac{\pi }{4\beta }} {L_\mathrm D}\sigma _0^2\exp \left( {\! - \frac{{{t^2}}}{{2\eta_\mathrm N }}} \right)} }\right) }\right|_{ - {B_\mathrm N}}^{Z - {B_\mathrm N}}\! \\
\left. {+F_\mathrm N \left({\frac{{{t^2}}}{{2\eta_\mathrm N }}}\right) } \right|_{ - {B_\mathrm N}}^0 - \left.{\left. {F_\mathrm N \left({\frac{{{t^2}}}{{2\eta_\mathrm N }}}\right) } \right|_0^{Z - {B_\mathrm N}}}\right]\\
 {\rm{ = }} C({\cal V}_\mathrm N)\times \left\{ {\frac{1}{{4{\beta}_\mathrm N }}\frac{\sqrt{\pi}}{D_\mathrm N} \left[ {{\rm{erf}}\left( {D_\mathrm N \left( {Z - {B_\mathrm N}} \right)} \right) + }\right.}\right. {\left.{{\rm{erf}}\left( {D_\mathrm N {B_\mathrm N}} \right)} \right]} - \\
 \left.{\frac{1}{2}\sqrt {\frac{\pi }{{\beta}_\mathrm N }} {L_\mathrm D}\sigma _0^2\left( {\text e ^{ -a_\mathrm N} \!-\! \text e ^{ -b_\mathrm N}} \right) + 2F_\text N(0)  - F_\text N(b_\mathrm N)  -  F_\text N(a_\mathrm N)}\right\}.
\end{array}
\end{equation}
\end{proof}
When $n$ goes to infinity, (16) is an accurate expression. However, finite items can guarantee the accuracy of approximation in practical use.
\end{prop}

\subsubsection{LOS for both RSU-user and RSU-IRS links}

In this case, $|{\mathrm h_{\mathrm{UI},i}}||{\mathrm h_{\mathrm{IR},i}}|$ is the product of two Rician variables and we denote the mean and variance of $G$ as ${\cal E}_\mathrm L$ and ${\cal V}_\mathrm L$, respectively, which can be expressed as
\begin{subequations}
\begin{equation}
{\cal E}_\mathrm L = M\sigma _\text m^2\frac{\pi }{2}{\text e^{ - K}}{\left[ {(1 + K){I_0}\left( {\frac{K}{2}} \right) + K{I_1}\left( {\frac{K}{2}} \right)} \right]^2},
\end{equation}
\begin{equation}
{\cal V}_\mathrm L =M\left[{ \left({2\sigma _\mathrm m^2 +{A^2}}\right)^2 - {\cal E}_\mathrm L^2}\right].
\end{equation}
\end{subequations}
In this case $|H|$ is a Rician variable and it is very difficult to derive the PDF of $Z'$. Therefore, we convert the problem into a double integration based on the PDF of $|H|$ and $G$.

In the derivation, the Taylor expansion of Bessel function is used, which is given by
\begin{equation}
{I_0}\left( x \right) \approx \sum\limits_{j = 0}^n  {\frac{1}{{{{\left( {j!} \right)}^2}}}{{\left( x \right)}^{2j}}}.
\end{equation}

With (19), we can deal with the distribution of $|H|$ and solve the double integration problem. The outage probability is given in Proposition 2.
\begin{prop}
When the RSU-user link and the RSU-IRS link are both LOS, the outage probability can be expressed as
\begin{equation}
  \begin{array}{l}
  {\mathrm P_{\mathrm {O,LL}}}(Z){\rm{ = }} \sum\limits_{i = 0}^n  {\sum\limits_{j = 0}^n  {{C_{\mathrm L,ij}}} }  \times \\
  \left\{ {{{E_\mathrm L}^{n_{ij}}}\frac{1}{2}\sqrt {\frac{\pi }{\mu_\mathrm L }} \left[ {\mathrm{erf}\left( {  \sqrt {\mu_\mathrm L}  M{{\cal E}_\mathrm L}} \right) - }\right.}\right.\mathrm{erf}\left.{\left( {\sqrt {\mu_\mathrm L}  E_\mathrm L} \right)} \right]  + \\
 \left.{\sum\limits_{k = 1}^{n_{ij}} {\mathop {\rm C}\nolimits_{n_{ij}}^k \!{{E_\mathrm L}^{n_{ij} - k}}} \left[{ \Gamma_k\left({{\mu_\mathrm L E}_\mathrm L^2}\right)\!+\!{{\left( { - 1} \right)}^k} \Gamma_k\left( {\mu_\mathrm L {{\left( {M{{\cal E}_\mathrm L}} \right)}^2}} \right) }\right]}\right\},
\end{array}
\end{equation}
where $n_{ij}=2i+2j+2$, $\mu_\mathrm L {\rm{ = }}({{2M{{\cal V}_\mathrm L}}})^{-1}$, $E_\mathrm L=M{{\cal E}_\mathrm L} - Z/L_\mathrm R$, $\Gamma_k(x)=({2{\mu_\mathrm L ^{{\gamma _k}}}})^{-1}\left({\Gamma \left( {{\gamma _k},0} \right)-\Gamma \left( {{\gamma _k},x} \right)}\right)$ and ${\gamma _k}{\rm{ = }}\frac{{k + 1}}{2}$.
${C_{\mathrm L,ij}}$ is given by (15) at the bottom of this page.
\end{prop}
\begin{proof}
We derive the outage probability as follows.
\begin{equation}
  \begin{array}{l}
\mathrm P\left[ {{L_\mathrm D}|H| + {L_\mathrm R}G < Z} \right] = \int_0^{\frac{Z}{{{L_\mathrm R}}}} {{f_{_G}}\left( g \right)\mathrm dg\int_0^{\frac{{Z - {L_\mathrm R}g}}{{{L_\mathrm D}}}} {{f_{_{\left| H \right|}}}\left( h \right)\mathrm dh} } \\
{\rm{ = }}\int_0^{\frac{Z}{{{L_\mathrm R}}}} \frac{1}{{\sqrt {2\pi } \sqrt {M{{\cal V}_\mathrm L}} }}{\rm{exp}}\left( { - \frac{{{{\left( {g - M{{\cal E}_\mathrm L}} \right)}^2}}}{{2M{{\cal V}_\mathrm L}}}} \right)\mathrm dg\\
\int_0^{\frac{{Z - {L_\mathrm R}g}}{{{L_\mathrm D}}}} {\frac{h}{{\sigma _\mathrm m^2}}\mathrm{exp}\left( { - \frac{{{A^2} + {h^2}}}{{2\sigma _\mathrm m^2}}} \right){I_0}\left( {\frac{{Ah}}{{\sigma _\mathrm m^2}}} \right)\mathrm dh}  \\
 = \frac{1}{{\sqrt {2\pi M{{\cal V}_\mathrm L}} }}\frac{1}{{\sigma _\mathrm m^2}}\mathrm{exp}\left( { - \frac{{{A^2}}}{{2\sigma _\mathrm m^2}}} \right)\int_0^{\frac{Z}{{{L_\mathrm R}}}} {{\rm{exp}}\left( { - \frac{{{{\left( {g - M{{\cal E}_\mathrm L}} \right)}^2}}}{{2M{{\cal V}_\mathrm L}}}} \right)\mathrm dg} \\
 \int_0^{\frac{{Z - {L_\mathrm R}g}}{{{L_\mathrm D}}}} {h \cdot \sum\limits_{i = 0}^\infty  {\frac{{{{\left( { - 1} \right)}^i}}}{{i!}}{{\left( {\frac{{{h^2}}}{{2\sigma _\mathrm m^2}}} \right)}^i}}  \sum\limits_{j = 0}^\infty  {\frac{1}{{{{\left( {j!} \right)}^2}}}{{\left( {\frac{{Ah}}{{2\sigma _\mathrm m^2}}} \right)}^{2j}}} \mathrm dh} \\
{\rm{ = }}\frac{1}{{\sqrt {2\pi M{{\cal V}_\mathrm L}} }}\frac{1}{{\sigma _\mathrm m^2}}\mathrm{exp}\left( { - \frac{{{A^2}}}{{2\sigma _\mathrm m^2}}} \right)\int_0^{\frac{Z}{{{L_\mathrm R}}}} {{\rm{exp}}\left( { - \frac{{{{\left( {g - M{{\cal E}_\mathrm L}} \right)}^2}}}{{2M{{\cal V}_\mathrm L}}}} \right)\mathrm dg} \\
\int_0^{\frac{{Z - {L_\mathrm R}g}}{{{L_\mathrm D}}}} {\sum\limits_{i = 0}^\infty  {\sum\limits_{j = 0}^\infty  {\frac{1}{{{2^{i + 2j}}\sigma _\mathrm m^{2i + 4j}}}\frac{{{{\left( { - 1} \right)}^i}{A^{2j}}}}{{i! {{\left( {j!} \right)}^2}}}} } {h^{2i + 2j + 1}}} dh\\
 = \sum\limits_{i = 0}^\infty  {\sum\limits_{j = 0}^\infty  {{C_{\mathrm L,ij}}} }
 \int_0^{\frac{Z}{{{L_\mathrm R}}}} {{\rm{exp}}\left( { - \frac{{{{\left( {g - M{{\cal E}_\mathrm L}} \right)}^2}}}{{2M{{\cal V}_\mathrm L}}}} \right){{\left( {g - \frac{Z}{{{L_\mathrm R}}}} \right)}^{2i + 2j + 2}}\mathrm dg}\\
{\rm{ = }}\sum\limits_{i = 0}^\infty  {\sum\limits_{j = 0}^\infty  {{C_{\mathrm L,ij}}} }  \int_0^{\frac{Z}{{{L_\mathrm R}}}} {{\rm{exp}}\left( { - \frac{{{{\left( {g - M{{\cal E}_\mathrm L}} \right)}^2}}}{{2M{{\cal V}_\mathrm L}}}} \right)}\times\\
{{\sum\limits_{k = 0}^{n_{ij}} {\mathop {\rm C}\nolimits_{n_{ij}}^k {{\left( {g - M{{\cal E}_\mathrm L}} \right)}^k}\left( {M{{\cal E}_\mathrm L} - \frac{Z}{{{L_\mathrm R}}}} \right)} }^{n_{ij} - k}}\mathrm dg \\
 = \sum\limits_{i = 0}^\infty  {\sum\limits_{j = 0}^\infty  {{C_{ij}}} }  \left[ {{{E_\mathrm L}^{n_{ij}}}\frac{1}{2}\sqrt {\frac{\pi }{\mu_{\mathrm L} }}\mathrm{erf}\left( {\sqrt {\mu_{\mathrm L}}  t} \right) - }\right.\\
 \left.{\left.{\sum\limits_{k = 1}^{n_{ij}} {\mathop {\rm C}\nolimits_{n_{ij}}^k {{{E_\mathrm L}}^{n_{ij} - k}}\frac{{\Gamma \left( {{\gamma _k},\mu_{\mathrm L} {t^2}} \right)}}{{2{\mu_{\mathrm L} ^{{\gamma _k}}}}}} } \right]} \right|_{ - M{{\cal E}_\mathrm L}}^{-\mu_{\mathrm L}}.
\end{array}
\end{equation}

Then, we can obtain (21).
\end{proof}

\subsubsection{LOS for the RSU-user link and NLOS for the RSU-IRS link}
In this case, $|H|$ is a Rician distributed variable and $|{\mathrm h_{\mathrm{UI},i}}||{\mathrm h_{\mathrm{IR},i}}|$ is the product of a Rician distributed variable and a Rayleigh distributed variable. The derivation process is similar to the second case, while the mean and variance of $G$ are ${\cal E}_\mathrm N$ and ${\cal V}_\mathrm N$, respectively.

\begin{prop}
When the RSU-user link is LOS and the RSU-IRS link is NLOS, the outage probability can be expressed as
\begin{equation}
  \begin{array}{l}
  {\mathrm P_\mathrm {O,LN}}(Z){\rm{ = }} \sum\limits_{i = 0}^n  {\sum\limits_{j = 0}^n  {{C_{\mathrm N,ij}}} }  \times \\
  \left\{ {{{E_\mathrm N}^{n{ij}}}\frac{1}{2}\sqrt {\frac{\pi }{\mu_\mathrm N }} \left[ {\mathrm{erf}\left( {  \sqrt {\mu_\mathrm N}  M{{\cal E}_\mathrm N}} \right) - }\right.}\right.\mathrm{erf}\left.{\left( {\sqrt {\mu_\mathrm N}  E_\mathrm N} \right)} \right]  + \\
 \left.{\sum\limits_{k = 1}^{n_{ij}} {\mathop {\rm C}\nolimits_{n_{ij}}^k \!{{E_\mathrm N}^{n_{ij}\! - \!k}}} \left[{ \Gamma_k\!\left({{\mu_\mathrm N E}_\mathrm N^2}\right)\!+\!{{\left( { \!- 1} \right)}^k} \Gamma_k\left( {\mu_\mathrm N {{\left( {M{{\cal E}_\mathrm N}} \right)}^2}} \right) }\right]}\right\},
\end{array}
\end{equation}
where ${C_{\mathrm N,ij}}\!=\!C_{\mathrm L,ij}\sqrt{{{\cal V}_\mathrm L}/{{\cal V}_\mathrm N}}$, \ $E_\mathrm N=M{{\cal E}_\mathrm N} \!-\! Z/L_\mathrm R$, \\$\mu_\mathrm N {\rm{ = }}({{2M{{\cal V}_\mathrm N}}})^{-1}$ and $\Gamma_k(x)\!=\!(2{\mu_\mathrm N ^{{\gamma _k}}})^{-1}\!\left({\Gamma \left( {{\gamma _k},0} \right)\!-\!\Gamma \left( {{\gamma _k},x} \right)}\right)$.
\end{prop}

\subsubsection{NLOS for RSU-user link and LOS for RSU-IRS link}
In this case $|H|$ is a Rayleigh distributed variable and $|{\mathrm h_{\mathrm{UI},i}}||{\mathrm h_{\mathrm{IR},i}}|$ is the product of two Rician distributed variables. The derivation process is similar to the first case, while the mean and variance of $G$ are ${\cal E}_\mathrm L$ and ${\cal V}_\mathrm L$, respectively.

\begin{prop}
When the RSU-user link is NLOS and the RSU-IRS link is LOS, the outage probability can be expressed as
\begin{equation}
  \begin{array}{l}
{\mathrm P_{\mathrm O,\text{NL}}}(Z){\rm{ = }} C({\cal V}_\mathrm L) \times \\
\left\{ {\frac{1}{{4\beta_\mathrm L }}\frac{\sqrt{\pi}}{D_\mathrm L} \left[ {{\rm{erf}}\left( {D_\mathrm L \left( {Z - B_\mathrm L} \right)} \right) + }\right.}\right. {\left.{{\rm{erf}}\left( {D_\mathrm L B_\mathrm L} \right)} \right]} - \\
 \left.{\frac{1}{2}\sqrt {\frac{\pi }{\beta_\mathrm L }} {L_\mathrm D}\sigma _0^2\left( {\text e ^{ -a_\mathrm L} \!-\! \text e ^{ -b_\mathrm L}} \right) + 2F_\text L(0)\!  -\! F_\text L(b_\mathrm L) \! - \! F_\text L(a_\mathrm L)}\right\},
\end{array}
\end{equation}
where $\beta_\mathrm L\!\!=\!\!\left({{{M{{\cal V}_\mathrm L}L_\mathrm R^2 \!+\! L_\mathrm D^2\sigma _0^2}}}\right)\!/\!\left({{{2M{{\cal V}_\mathrm L}L_\mathrm R^2\sigma _0^2}}}\right)$,$b_\mathrm L\!\!=\!\!B_\mathrm L^2/(2A_\mathrm L)$\\$D_\mathrm L=\sqrt {({{A_\mathrm L  + 2\beta_\mathrm L L_\mathrm D^2\sigma _0^4}})/({{2{A_\mathrm L ^2}}})}$, $a_\mathrm L\!=\!\left( {Z\! -\! {B_\mathrm L}} \right)^2/(2A_\mathrm L)$,\\ $B_\mathrm L\!=\!{L_\mathrm R}M{{\cal E}_\mathrm L}$, with $A_\mathrm L\!=\!M{{\cal V}_\mathrm L}L_\mathrm R^2\! +\! L_\mathrm D^2\sigma _0^2$. $F_\text L(x)$is given by (16) at the bottom of this page.
\end{prop}

Considering the vehicle distribution and summarizing results in (17), (21), (23) and (24), the overall outage probability is given in Theorem 1.
\begin{thm}
The outage probability of the system described in Section II is given by
\begin{equation}
\begin{split}
{\mathrm P_\mathrm O}(Z) \!= &\mathrm P_\mathrm{O,NN}(Z) p_\mathrm{IR,N} p_\mathrm{UR,N} \!+\! \mathrm P_\mathrm{O,LL}(Z) p_\mathrm{IR,L} p_\mathrm{UR,L}\! +\! \\
&\mathrm P_\mathrm{O,LN}(Z) p_\mathrm{IR,N} p_\mathrm{UR,L}\!+\!
\mathrm P_\mathrm{O,NL}(Z) p_\mathrm{IR,L} p_\mathrm{UR,N}
\end{split}
\end{equation}
where $\mathrm P_\mathrm{O,NN}(Z)$, $\mathrm P_\mathrm{O,LL}(Z)$, $\mathrm P_\mathrm{O,LN}(Z)$ and $\mathrm P_\mathrm{O,NL}(Z)$ are given in (17), (21), (23) and (24), respectively.
\end{thm}

From Theorem 1, an accurate approximation of the outage probability is provided. The accuracy is related to the series order $n$ and the approximation would be more accurate when $n$ is large. In practice, the approximation performance would be very accurate when $n$ is around 10.

\subsection{Analysis based on CLT}

The expression for outage probability in Theorem 1 is accurate but very complicated. Therefore, we provide a simplified expression based on CLT. The basic idea is to use a Gaussian variable to approximate $Z'$.

Since $Z'$ is the sum of $M+1$ independent, uniform but nonidentical random variables, $Z'$ follows the Gaussian distribution when $M$ is large, according to CLT\cite{07350}.
Denote its mean and variance as $\varepsilon$ and $\nu$, which are expressed as
\begin{subequations}
\begin{equation}
\varepsilon = {L_\text D}{\cal E}_{|H|}+{L_\text R}{\cal E}_{G}
\end{equation}
\begin{equation}
\nu = {L_\text D^2}{\cal{ V}}_{|H|} + L_\text R^2 {\cal{ V}}_G.
\end{equation}
\end{subequations}
where ${\cal E}_{|H|}$, ${\cal{ V}}_{|H|}$, ${\cal E}_{G}$ and ${\cal{ V}}_G$ are the mean and variance of $|H|$ and $G$, respectively. ${\cal E}_{G}$ and ${\cal{ V}}_G$ in different conditions have been provided in (12) and (18), which are expressed as ${\cal E}_\text{N}$, ${\cal{ V}}_\text{N}$ and ${\cal E}_\text{L}$, ${\cal{ V}}_\text{L}$. For $|H|$, the mean and variance when the RSU-user link is LOS or NLOS are given, respectively, by ${\cal E}_{|H|,\mathrm L}= \sqrt {\pi/2} {\sigma _\mathrm m}{\text e^{ - K/2}}\left[ {(1 + K){I_0}\left( {K/2} \right) + K{I_1}\left( {K/2} \right)} \right]$,\\ ${\cal{ V}}_{|H|,\mathrm L}\!\!=\!\!2\sigma _{\rm{m}}^2 \!+\! {A^2}$ and ${\cal E}_{|H|,\mathrm N}\!\!=\!\!\sqrt {\pi\!/2} {\sigma _0}$, ${\cal{ V}}_{|H|,\mathrm N}\!=\!\left( {2{\rm{ \! -\! }}\pi/2} \right)\sigma _0^2$. By integration over $Z'$ and assuming that it is Gaussian distributed, the approximated outage probability is obtained in the following theorem.

\begin{thm}
Based on CLT, a simplified approximated expression of outage probability is given by
\begin{equation}
  \begin{array}{l}
\mathrm {P_O}(Z) = \frac{1}{2}\mathrm{erf}\left( {\sqrt {\frac{1}{{2\nu }}} \left( {Z - \varepsilon } \right)} \right) + \frac{1}{2}\mathrm{erf}\left( {\sqrt {\frac{1}{{2\nu }}} \varepsilon } \right).
\end{array}
\end{equation}
where $Z$ is the modified threshold in (11).
\end{thm}

\subsection{Discussion}

To gain more insights, we analyze the impact of the vehicle density and the number of IRS elements on the outage probability.
\newtheorem{cor}{Corollary}
\begin{cor}
The outage probability is a monotonically increasing function of $\lambda _\mathrm{O}$ or $\lambda _\mathrm{U}$. With higher vehicle density, the probability of channel blockage increases, leading to the increase of outage probability.
\end{cor}
\begin{proof}
For any fixed $Z$, use $\mathrm P_\mathrm O$, $\mathrm P_\text{O,NN}$, $\mathrm P_\text{O,LL}$, $\mathrm P_\text{O,LN}$ and $\mathrm P_\text{O,NL}$ to represent $\mathrm P_\mathrm O(Z)$, $\mathrm P_\mathrm{O,NN}(Z)$, $\mathrm P_\mathrm{O,LL}(Z)$, $\mathrm P_\mathrm{O,LN}(Z)$ and $\mathrm P_\mathrm{O,NL}(Z)$ for simplicity. According to (25), taking the first-order derivation of $\mathrm P_\mathrm O$ with respect to $\lambda _\mathrm{O}$, we obtain
\begin{equation}
\begin{aligned}
\frac{\mathrm{d} \mathrm P_\mathrm O }{\mathrm{d} \lambda _\mathrm{O}}=&
\left( {\mathrm P_\text{O,NN}- \mathrm P_\text{O,NL} } \right) \left( { \mathrm e^{ - ({\lambda _\mathrm{O}+\lambda _\mathrm{U}})\tau } - \mathrm e^{ - ({2\lambda _\mathrm{O}+\lambda _\mathrm{U}})\tau } }\right)\tau + \\
&\left( {\mathrm P_\text{O,NN}- \mathrm P_\text{O,LN} } \right) \left( {\mathrm e^{-\lambda _\mathrm{O}\tau} - \mathrm e^{ - ({2\lambda _\mathrm{O}+\lambda _\mathrm{U}})\tau }}\right)\tau + \\
&\left( {\mathrm P_\text{O,LN} + \mathrm P_\text{O,NL} - 2\mathrm P_\text{O,LL} } \right)\mathrm e^{ - ({2\lambda _\mathrm{O}+\lambda _\mathrm{U}})\tau }\tau.
\end{aligned}
\end{equation}

Deriving the expressions of $\mathrm P_\text{O,NN}-\mathrm P_\text{O,NL}$, $\mathrm P_\text{O,NN}-\mathrm P_\text{O,LN}$ and $\mathrm P_\text{O,LN} + \mathrm P_\text{O,NL} - 2\mathrm P_\text{O,LL}$, and taking the first-order derivation of these expressions with respect to $Z$, we can prove that $\mathrm P_\text{O,NN}>\mathrm P_\text{O,LN}>\mathrm P_\text{O,LL}$ and $\mathrm P_\text{O,NN}>\mathrm P_\text{O,NL}>\mathrm P_\text{O,LL}$. Then we have $\frac{\mathrm{d} \mathrm P_\mathrm O }{\mathrm{d} \lambda _\mathrm{O}}>0$. Similarly, we have $\frac{\mathrm{d} \mathrm P_\mathrm O }{\mathrm{d} \lambda _\mathrm{U}}>0$. Thus, the proof is completed.

\end{proof}
\begin{cor}
According to (27), when $M\!=\!0$, the IRS is out of service and the outage probability reaches a peak value. However, when $M\!\to\! \infty$, the outage probability will approach 0.
\end{cor}
\begin{proof}
With the increase of $M$, the values of $\varepsilon$ and $\nu$ will decrease.
In (26), when the value of $M \!\to\! \infty$, we have $\varepsilon \to \infty$ and $\nu \to \infty$. Since $\varepsilon^2$ and $\nu$ are both second order infinity of $M$, we have $\sqrt {\frac{1}{{2\nu }}} \varepsilon\to C_0$ and $\sqrt {\frac{1}{{2\nu }}} Z \to 0$, where $C_0$ is a constant value. Therefore, $\mathrm {P_O}(Z)$ will approach 0.
\end{proof}

Note that in reality, it is impossible for $M$ to be infinity and the outage probability would never become 0. In general, larger value of $M$ will reduce the outage probability.

\section{Numerical results}

In this section, we carry out Monte-Carlo simulations to evaluate the performance of IRS-assisted vehicular communication systems and the accuracy of analysis in previous sections. The simulation parameters are summarized
in Table I.

\begin{table}[h]
\centering
\caption{Main simulation parameters}
\begin{threeparttable}
\begin{tabular}{p{2.3cm}<{\centering}p{2.2cm}<{\centering}p{1.20cm}<{\centering}p{1.5cm}<{\centering}}
\toprule[2pt]

Parameter&Value&Parameter&Value\\
\midrule[1pt]
Lane width & 4 m&$\mathrm{\alpha(z) _L}$ & 2.8\cite{TassiModeling}\\

Vehicle length $\tau$& 5 m&$\mathrm{\alpha(z) _N}$ & 4\cite{Rappaport2013Broadband}\\

Carrier frequency $f$ & 28 GHz&$\sigma _\mathrm n^2$ & -39dBm\cite{TassiModeling}\\

\(\mathrm{C_L, C_N}\) & $ - 20{\rm log}\left( {4\pi f/c} \right)$ \tnote{1} & $P_\mathrm{tx}$ & 27dBm\cite{TassiModeling}\\
\bottomrule[2pt]
\end{tabular}
\begin{tablenotes}
\item[1] It is the free space path loss in dB at a distance of 1m and $c$ is the speed of light\cite{Shu2016Millimeter}.
\end{tablenotes}
\end{threeparttable}
\end{table}

As shown in Fig. 3, the width of each lane is 4m and the vehicle length is 5m. The coordinates of RSU and IRS are (10,0) and (22,8), respectively. The transmit power of RSU is 27dBm. Parameters related to pass loss are given in Table I.
\begin{figure}
    \centering
    \includegraphics[scale=0.37]{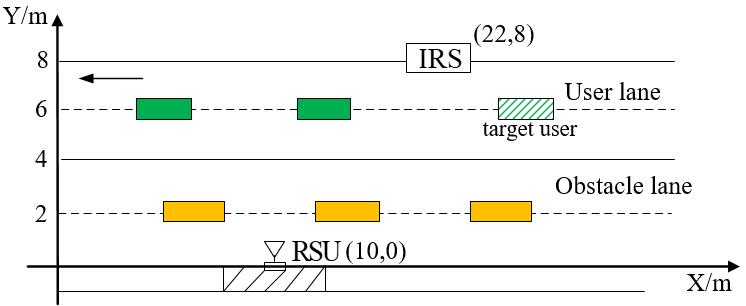}
    \caption{An illustration of the simulation setup}
\end{figure}

In Fig. 4, the outage probabilities calculated using SEA and CLA are compared, where the number of IRS elements is fixed at 100 and ${\lambda _{\rm{u}}}$ and ${\lambda _{\rm{o}}}$ are both 1/10, which means the average distance between two adjacent vehicles is 5m. This implies a very crowded scenario, where outage would be severe. As can be seen in Fig. 4, the outage probabilities differ greatly at different locations of the target user. For a given threshold $t$, the SEA and CLA proposed in this paper are very accurate, the results of which almost coincide with simulation results. SEA with $n\!=\!10$ is the most accurate, but the performance of $n\!=\!3$ is fine as well. CLA has excellent performance and its expression is concise, which makes it very helpful for analysis.

\begin{figure}
    \centering
    \includegraphics[scale=0.53]{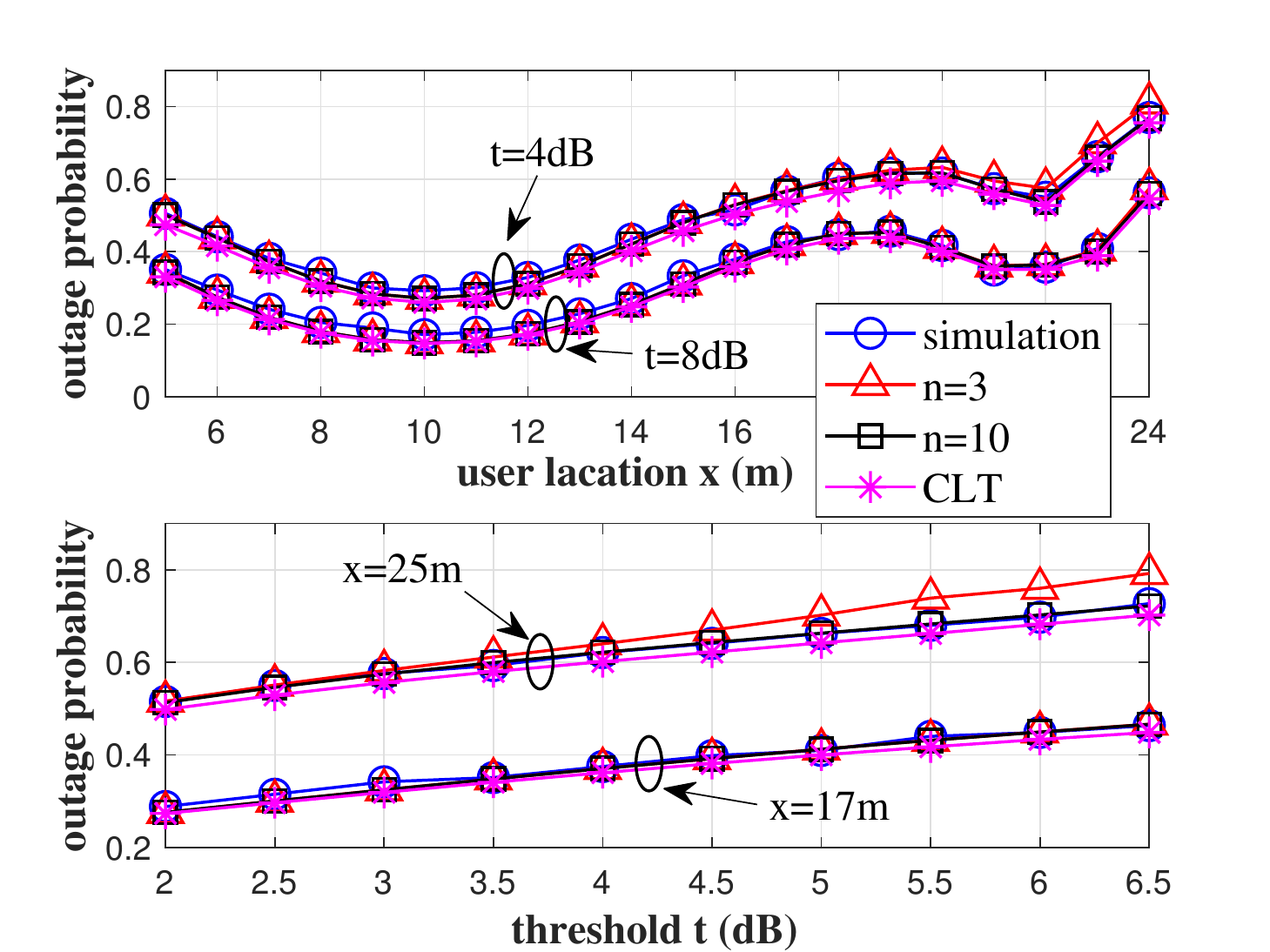}
    \caption{The comparison of simulation results and analytical results based on SEA and CLA}
\end{figure}

Fig. 5 depicts the relationship between outage probability and the user location with different number of IRS elements within the range of 0 to 500. Each row in Fig. 5 indicates the outage probability of the whole area for a specific value of $M$. It can be seen that the outage probability decreases with $M$. With large value of $M$, the outage probabilities in the vicinity of the IRS decrease significantly. When $M=500$, the outage probability near the IRS is even smaller than that near the RSU. The reason is that large IRS can greatly improve the coverage performance of the RSU by reflecting signals. Although mmWave vehicular communication has a small coverage range and suffers from severe blockage, IRS could help to increase the outage performance significantly. Compared with $M=0$ (no IRS used), the range for $\mathrm{P_O}<0.3$ increases for 3 times when $M=500$. Note that the range where IRS could work effectively is quite small. Therefore, it is feasible to deploy multiple collaborative IRSs along the roadside to improve the coverage performance of all regions.

\begin{figure}
    \centering
    \includegraphics[width=8cm,height=4.85cm]{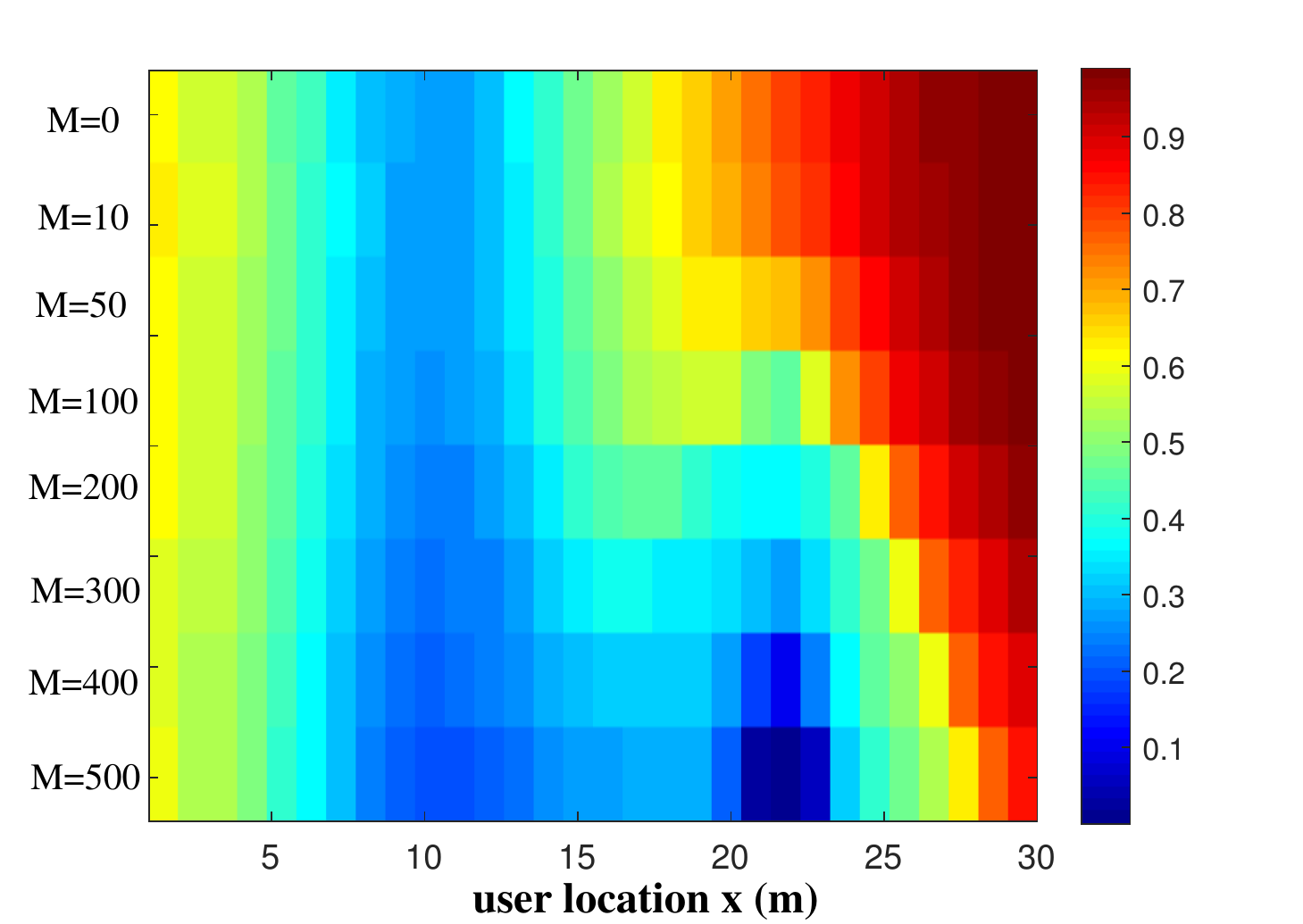}
    \caption{Outage probability versus different number of IRS elements, where $t=4\mathrm{dB}$}
\end{figure}

To evaluate the influence of vehicle distance $d$ on the outage probability, we assume that the average vehicle distance of the two lanes is identical. The results with  $d=\mathrm{5m}$ ($\tau$), $d=\mathrm{25m}$ (5$\tau$) and $d=\mathrm{50m}$ (10$\tau$) are depicted in Fig. 6.
As the user location changes, the outage probability reaches its minimum at $x=22\mathrm m$, where the IRS locates.
It can be concluded that the outage probability decreases with respect to $d$ due to the reduced probability of blockage.
\begin{figure}
    \centering
    \includegraphics[scale=0.49]{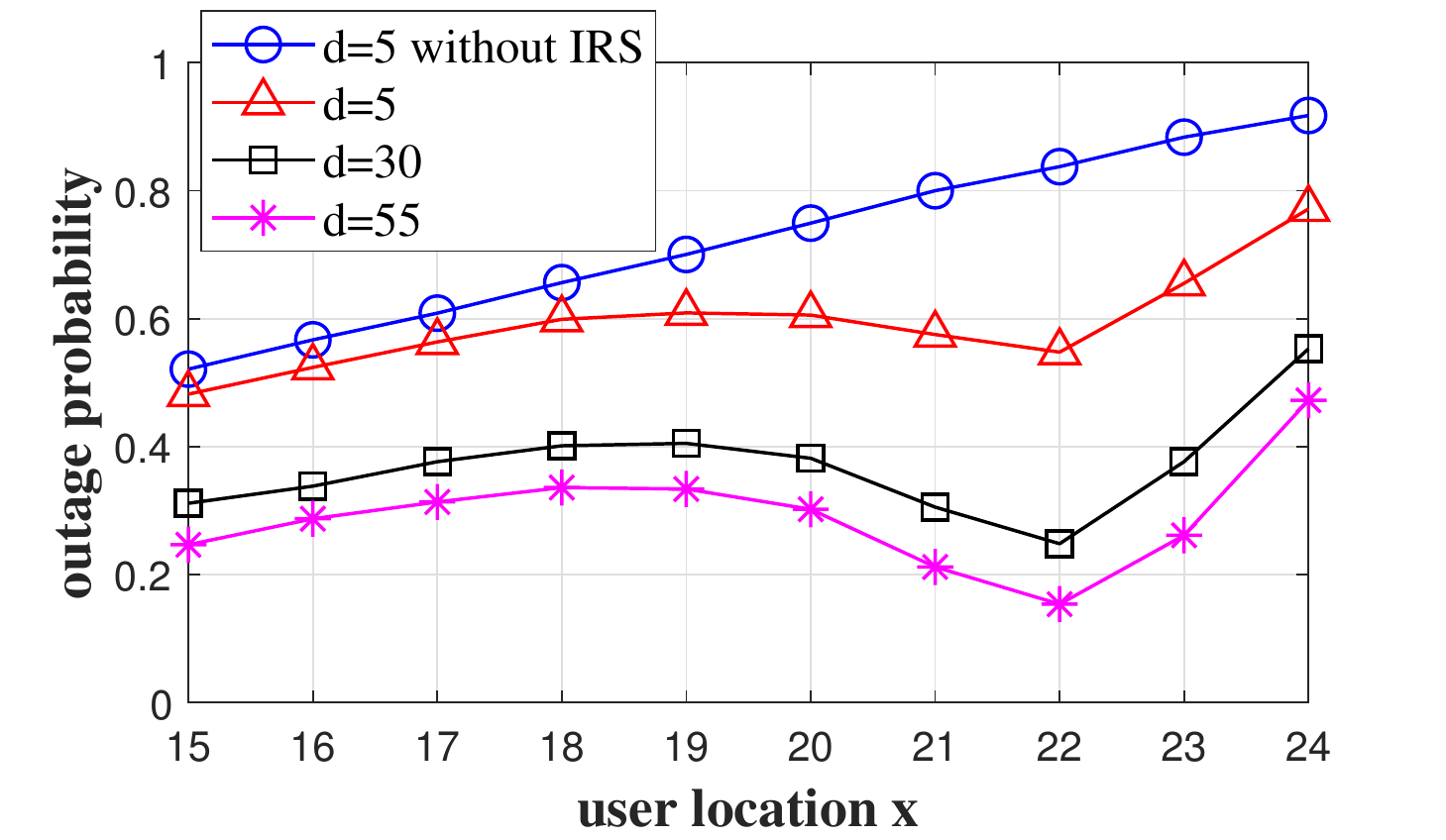}
    \caption{Outage probability versus different vehicle distance}
\end{figure}

\section{Conclusion}
In this work, we studied the outage probability for IRS-assisted vehicular communication systems. Accurate approximate expressions of outage probability were provided by means of SEA and CLA. Numerical results validate the analysis. It is shown that the outage probability can be reduced significantly by using large IRS even for very crowded road.

\appendices

\section*{Acknowledgment}

This work was supported in part by the Key Research \& Development Plan of Jiangsu Province (No. BE2018108), National Nature Science Foundation of China (Nos.61701198 \& 61772243), Nature Science Foundation of Jiangsu Province (No. BK20170557) and Young Talent Project of Jiangsu University.

\bibliographystyle{IEEEtran}
\bibliography{IEEEabrv,My}

\begin{thebibliography}{10}
\providecommand{\url}[1]{#1}
\csname url@samestyle\endcsname
\providecommand{\newblock}{\relax}
\providecommand{\bibinfo}[2]{#2}
\providecommand{\BIBentrySTDinterwordspacing}{\spaceskip=0pt\relax}
\providecommand{\BIBentryALTinterwordstretchfactor}{4}
\providecommand{\BIBentryALTinterwordspacing}{\spaceskip=\fontdimen2\font plus
\BIBentryALTinterwordstretchfactor\fontdimen3\font minus
  \fontdimen4\font\relax}
\providecommand{\BIBforeignlanguage}[2]{{%
\expandafter\ifx\csname l@#1\endcsname\relax
\typeout{** WARNING: IEEEtran.bst: No hyphenation pattern has been}%
\typeout{** loaded for the language `#1'. Using the pattern for}%
\typeout{** the default language instead.}%
\else
\language=\csname l@#1\endcsname
\fi
#2}}
\providecommand{\BIBdecl}{\relax}
\BIBdecl

\bibitem{8594703}
F.~{Jameel}, S.~{Wyne}, S.~J. {Nawaz}, and Z.~{Chang}, ``Propagation channels
  for mmwave vehicular communications: State-of-the-art and future research
  directions,'' \emph{IEEE Wireless Communications}, vol.~26, no.~1, pp.
  144--150, February 2019.

\bibitem{8539687}
Z.~{Sheng}, A.~{Pressas}, V.~{Ocheri}, F.~{Ali}, R.~{Rudd}, and M.~{Nekovee},
  ``Intelligent 5g vehicular networks: An integration of dsrc and mmwave
  communications,'' in \emph{2018 International Conference on Information and
  Communication Technology Convergence (ICTC)}, 2018, pp. 571--576.

\bibitem{8617303}
T.~{Shimizu}, V.~{Va}, G.~{Bansal}, and R.~W. {Heath}, ``Millimeter wave v2x
  communications: Use cases and design considerations of beam management,'' in
  \emph{2018 Asia-Pacific Microwave Conference (APMC)}, 2018, pp. 183--185.

\bibitem{pan2019multicell}
C.~Pan, H.~Ren, K.~Wang, W.~Xu, M.~Elkashlan, A.~Nallanathan, and L.~Hanzo,
  ``Multicell mimo communication relying on intelligent reflecting surface,''
  \emph{[Online] https://arxiv.org/abs/1907.10864}, 2019.

\bibitem{8796365}
E.~{Basar}, M.~{Di Renzo}, J.~{De Rosny}, M.~{Debbah}, M.~{Alouini}, and
  R.~{Zhang}, ``Wireless communications through reconfigurable intelligent
  surfaces,'' \emph{IEEE Access}, vol.~7, pp. 116\,753--116\,773, 2019.

\bibitem{8879620}
Z.~{He} and X.~{Yuan}, ``Cascaded channel estimation for large intelligent
  metasurface assisted massive mimo,'' \emph{IEEE Wireless Communications
  Letters}, vol.~9, no.~2, pp. 210--214, Feb 2020.

\bibitem{03272}
Z.~Wang, L.~Liu, and S.~Cui, ``Intelligent reflecting surface-enhanced ofdm:
  Channel estimation and reflection optimization,'' \emph{arXiv preprint
  arXiv:1909.03272}.

\bibitem{pan2019inte}
C.~Pan, H.~Ren, K.~Wang, M.~Elkashlan, A.~Nallanathan, J.~Wang, and L.~Hanzo,
  ``Intelligent reflecting surface enhanced mimo broadcasting for simultaneous
  wireless information and power transfer,'' \emph{arXiv preprint
  arXiv:1908.04863}, 2019.

\bibitem{8683145}
Q.~{Wu} and R.~{Zhang}, ``Beamforming optimization for intelligent reflecting
  surface with discrete phase shifts,'' in \emph{ICASSP 2019 - 2019 IEEE
  International Conference on Acoustics, Speech and Signal Processing
  (ICASSP)}, May 2019, pp. 7830--7833.

\bibitem{hong2019}
S.~Hong, C.~Pan, H.~Ren, K.~Wang, and A.~Nallanathan, ``Artificial-noise-aided
  secure mimo wireless communications via intelligent reflecting surface,''
  \emph{arxiv preprint arXiv:2002.07063}, 2020.

\bibitem{12183}
A.~U. Makarfi, K.~M. Rabie, O.~Kaiwartya, X.~Li, and R.~Kharel, ``Physical
  layer security in vehicular networks with reconfigurable intelligent
  surfaces,'' \emph{arXiv preprint arXiv:1912.12183}, 2019.

\bibitem{TassiModeling}
A.~{Tassi}, M.~{Egan}, R.~J. {Piechocki}, and A.~{Nix}, ``Modeling and design
  of millimeter-wave networks for highway vehicular communication,'' \emph{IEEE
  Transactions on Vehicular Technology}, vol.~66, no.~12, pp. 10\,676--10\,691,
  2017.

\bibitem{07350}
Y.~I., A.~Uyrus, E.~Basar, and I.~F. Akyildiz, ``Propagation modeling and
  analysis of reconfigurable intelligent surfaces for indoor and outdoor
  applications in 6g wireless systems,'' \emph{arXiv preprint
  arXiv:1912.07350}, 2019.

\bibitem{Rappaport2013Broadband}
T.~S. Rappaport, F.~Gutierrez, E.~Ben-Dor, J.~N. Murdock, and J.~I. Tamir,
  ``Broadband millimeter-wave propagation measurements and models using
  adaptive-beam antennas for outdoor urban cellular communications,''
  \emph{IEEE Transactions on Antennas and Propagation}, vol.~61, no.~4, pp.
  1850--1859, 2013.

\bibitem{Shu2016Millimeter}
S.~Shu, G.~R. Maccartney, and T.~S. Rappaport, ``Millimeter-wave
  distance-dependent large-scale propagation measurements and path loss models
  for outdoor and indoor 5g systems,'' in \emph{2016 10th European Conference
  on Antennas and Propagation (EuCAP)}, 2016.

\end{thebibliography}

\end{document}